\documentclass[11pt,letterpaper]{article}

\usepackage[utf8]{inputenc}
\usepackage{xcolor,color}
\usepackage{url}
\usepackage{authblk}

\usepackage{hyperref}
\hypersetup{colorlinks,linkcolor=black,filecolor=black,urlcolor=black,citecolor=black}

\usepackage{pstricks}
\usepackage{pdftricks}
\usepackage{amsmath,amssymb,amsthm,amsfonts}

\theoremstyle{plain}
\newtheorem{theorem}{Theorem}

\theoremstyle{Definition}

\usepackage[shortlabels]{enumitem} 

\newcommand{\cB}{\mathcal{B}}
\newcommand{\dist}{\mathrm{dist}}
\newcommand{\diam}{\mathrm{diam}}

\newdimen\omsq  \omsq=10pt
\newdimen\omrule    \omrule=1pt
\newdimen\omint

\newif\ifvth    \newif\ifhth    \newif\ifomblank

\def\OM#1{%
    \omint=\omsq    \advance\omint-\omrule
    \ifx\relax#1%
    \else
      \ifx\\#1 \newline\null \hthtrue \ifvth\vthfalse\else\vskip-\omrule\vthtrue\fi
      \else%
        \ifx .#1\hskip\ifhth \omrule\else \omint\fi
        \else%
          \ifx +#1\def\colour{black}\fi%
          \ifx -#1\def\colour{black}\fi%
          \ifx |#1\def\colour{black}\fi%
          \ifx 2#1\def\colour{lightgray}\fi%
          \ifx 3#1\def\colour{gray}\fi%
          \ifx 4#1\def\colour{darkgray}\fi%
          \ifx @#1\def\colour{black}\fi%
          \ifx r#1\def\colour{red}\fi%
          \ifx g#1\def\colour{green}\fi%
          \ifx l#1\def\colour{lime}\fi%
          \ifx o#1\def\colour{olive}\fi%
          \ifx O#1\def\colour{orange}\fi%
          \ifx b#1\def\colour{blue}\fi%
          \ifx t#1\def\colour{teal}\fi%
          \ifx y#1\def\colour{yellow}\fi%
          \ifx m#1\def\colour{magenta}\fi%
          \ifx c#1\def\colour{cyan}\fi%
          \ifx p#1\def\colour{pink}\fi%
          \ifx P#1\def\colour{purple}\fi%
          \ifx w#1\def\colour{white}\fi%
          \textcolor{\colour}{\rule{\ifhth\omrule\else\omsq\fi}{\ifvth\omrule\else\omsq\fi}}%
          \ifhth\else\hskip -\omrule\fi%
        \fi%
        \ifhth\hthfalse\else\hthtrue\fi%
      \fi%
    \expandafter\OM%
    \fi}

\makeatother

\title{The bright side of simple heuristics for the TSP}

\author{Alan Frieze\thanks{Research Supported in part by NSF grant DMS1952285. email: frieze@cmu.edu}}
\author{Wesley Pegden\thanks{Research Supported in part by NSF grant DMS1700365. email: wes@math.cmu.edu}}
\affil{Department of Mathematical Sciences, Carnegie Mellon University}

\newcommand{\cN}{\mathcal{N}}
\newcommand{\Ndp}{\cN_\delta^{\mathrm{pack}}}

\newcommand{\g}{\gamma}

\newcommand{\cM}{\mathcal{M}}

\newcommand{\cE}{\mathcal{E}}

\def\E{\mathbb{E}}

\def\cE{\mathcal E}
\def\real{\mathbf{R}}
\newcommand{\bfrac}[2]{\left(\frac{#1}{#2}\right)}
\begin{document}
\maketitle
\begin{abstract}
The greedy and nearest-neighbor TSP heuristics can both have $\log n$ approximation factors from optimal in worst case, even just for $n$ points in Euclidean space.  In this note, we show that this approximation factor is only realized when the optimal tour is unusually short.  In particular, for points from any fixed $d$-Ahlfor's regular metric space (which includes any $d$-manifold like the $d$-cube $[0,1]^d$ in the case $d$ is an integer but also fractals of dimension $d$ when $d$ is real-valued), our results imply that the greedy and nearest-neighbor heuristics have \emph{additive} errors from optimal on the order of the \emph{optimal} tour length through \emph{random} points in the same space, for $d>1$.
\end{abstract}

\section{Introduction}
Papadimitriou \cite{Papa} showed that finding an optimum Traveling Salespesron Tour is NP-hard even for points in Euclidean space, while Arora \cite{Aro} and Mithcell \cite{Mit} give polynomial-time approximation schemes for the Euclidean TSP.   In practice these have resisted efficient implementations, and in practice, Euclidean TSP approximation still leans heavily on heuristics which are not known to be asymptotically optimal.  For metric TSP, Christofides algorithm achieves an approximation ratio of 1.5, which saw slight improvement with the recent breakthrough of Karlin, Klein, Gharan, and Shayan \cite{Kar}.

Perhaps the simplest heuristics to find a tour through $n$ points are the Nearest Neighbor heuristic, which grows (and in the end, closes) a path by jumping at each step to the nearest unvisited point, and the Greedy heuristic, which at each step chooses the shortest available edge which would not create any vertices of degree 3 or close a cycle unless on the $n$th step.  For $n$ points in an arbitrary metric space, each of these heuristics is known to give a tour within $\log n$ of optimal \cite{greedy,nearest}, and examples are known which realize these approximation ratios, even just in Euclidean space. But our main result implies that for $n$ points in the unit square whose optimal tour has length $\Omega(\sqrt{n})$ (as is the typical case), the Greedy and Nearest Neighbor heuristics will both return a tour whose length is within a constant factor of optimal.

We will prove our results not just for full-dimensional Euclidean space but for any sufficiently regular metric space with dimension $d>1$; the point of this generality is to emphasize that for greedy or nearest-neighbor algorithms to have poor approximation ratios on some input, it is really necessary that the the input admits an unexpectedly short tour given the space its points are taken from, rather than, say, just because the input was actually chosen from a lower dimensional subset of the space than expected.

\bigskip

A metric space $\cM$ equipped with a measure $\mu$ is $d$-\emph{Ahlfor's regular} if there are constants $C,D$ so that 
\begin{equation}\label{e.ahlfors}
Cr^d\leq \mu(B(p,r))\leq Dr^d
\end{equation}
for all $p\in \cM$ and $0<r\leq \diam (\cM)$.  Here $B(p,r)$ is the ball of radius $r$ centred at $p$. Simple examples of regular metric spaces include subspaces of Euclidean space like unit cubes under Lebesgue measure (having integer dimensions) or fractals like the Sierpinski gasket under the Hausdorff measure (having intermediate dimensions)---for example, the metric space induced in Euclidean space by any fractal generated by an iterative function system satisfying the open set condition is Ahlfor's regular for some $d$, for the Hausdorff measure of appropriate dimension.

We will prove the following about optimal TSP tours in Ahlfor's-regular spaces:
\begin{theorem}\label{t.lb}
Suppose $x_1,x_2,\dots$ is a sequence of i.i.d points drawn from a $d$-Ahlfor's regular probability measure on the metric space $\cM$.  Then there exists a constant $C$ so that for $X_n=\{x_1,\dots,x_n\}$, the length of the optimal tour through $X_n$ has length at least $Cn^{1-\frac 1 d}$ for all sufficiently large $n$, with probability 1.
\end{theorem}

Our main result for the nearest-neighbor and greedy heuristics is then the following:
\begin{theorem}\label{t.NN}
If the bounded metric space $\cM$ admits a $d$-Ahlfor's regular measure, then there is a constant $C$ and an $n_0$ such that for any $n$ points in $\cM$ with $n\geq n_0$, the nearest-neighbor and greedy algorithms produce a tour of length at most $Cn^{1-\frac{1}{d}}$.
\end{theorem}


\section{Proofs}
\begin{proof}[Proof of Theorem \ref{t.lb}]
Let $D$ be the constant from \eqref{e.ahlfors} guaranteed to exist for $(\cM,\mu)$.  Let $r=\left(\frac{1}{Dn}\right)^{1/d}$.  For any fixed $i$, let $Z_i$ be the indicator for the event $\mathcal E_i$ that $x_i$ is the unique point from $X_n$ in $B(x_i,r)$. Then 
\[
\Pr(\cE_i)\geq (1-Dr^d)^{n-1}\geq e^{-1}.
\]
Let $Z=Z_1+\cdots+Z_n$. Thus $\E(Z)\geq e^{-1}n$. Let $\cB$ be the event that there exists $i$ such that $B(x_i,r)$ contains more than $\log^2n$ points from $X_n$ other than $x_i$. Then\
\[
\Pr(\cB)\leq n\Pr(Bin(n,Dr^d)\geq \gamma)\leq \binom{n}{\gamma}(Dr^d)^\gamma\leq \bfrac{e}{\gamma}^\gamma\leq n^{-\log n}.
\]
If $\cB$ does not occur then changing the value of one $x_i$ only changes the value of $Z$ by at most $\log^2n$. If $\cB$ does occur then $Z$ could change by at most $n$. We will now use Warnke's {\em Typical bounded differences inequality} \cite{War} to show that $Z$ is concentrated around its mean. 
\begin{theorem}[Warnke]\label{Wcon}
Let $X=(X_1,\ldots,X_N)$ be a family of independent random variables with $X_k$ taking values in a set $\Lambda_k$. Let $\Gamma\subseteq \prod_{j\in [N]}\Lambda_j$ be an event and assume that the function $f:\prod_{j\in [N]}\Lambda_j\to \real$ satisfies the typical Lipschitz condition: there are numbers $c_k,k\in[N]$ and $d_k,k\in [N]$ such that whenever $x,y$ differ only in the $k$th coordinate, we have
\[
|f(x)-f(y)|\leq \begin{cases}c_k&\text{if }x\in \Gamma.\\d_k&\text{otherwise}.\end{cases}
\]
Then for all numbers $\g_k,k\in [N]$ with $\g_k\in (0,1)$, 
\begin{multline*}
\Pr(|f(X)-\E(f(X))|\geq t)\leq \\
2\exp\left\{-\frac{t^2}{2\sum_{k\in [N]}(c_k+\gamma_k(d_k-c_k))^2}\right\}+\Pr(X\notin\Gamma)\sum_{k\in[N]}
\gamma_k^{-1}.
\end{multline*}
\end{theorem}
We will apply this theorem with $f=Z,N=n,X=\{x_1,\ldots,x_n\},\Gamma=\cB^c$ and $c_k=\gamma,d_k=n,\g_k=\log^2n,\gamma_k=n^{-2}$ for $k\in[n]$. This yilelds
\[
\Pr(Z\leq \E(Z)-n^{2/3})\leq 2\exp\left\{-\frac{n^{4/3}}{n(\log^2n+1)^2}\right\}+n^{3-\log n}=o(1).
\]
So, w.h.p. there are at least $n/3$ of the $x_i$ that are at least $r$ from their nearest neighbor. Theorem \ref{t.lb} follows immediately.
\end{proof}
\begin{proof}[Proof of Theorem \ref{t.NN}]

Consider any nearest-neighbor or greedy tour 
$x_1,\dots,x_n$
through the point-set $X=\{x_1,\dots,x_n\} \in \cM$.  We define a sequence of open balls $B_1,\dots,B_{n-1}$, where $B_i$ is centered at $x_i$ and has radius $\dist(x_i,x_{i+1})$.  Observe that when the edge from $\{x_i,x_{i+1}\}$ is selected, there can be no other vertices $x_j$ which would be available for selection but are closer to $x_i$ than $\dist(x_i,x_{i+1})$. This implies that the family $\cB=\{B_i\}$ has the following property: 
\smallskip

\noindent \textbf{($\star$)} For any distinct balls $B_i,B_j\in \cB$, we have either that the $B_i$ doesn't contain the center of $B_j$ or that $B_j$ doesn't contain the center of $B_i$ (according to whether $i<j$ or $j<i$, respectively).

\smallskip


Now we partition $\cB$ into sets $\cB_1,\cB_2,\dots$, where each $\cB_j$ consists of every ball $D\in \cB$ whose radius $r$ satisfies $\frac 1 {2^j} <r\leq \frac 1 {2^{j-1}}$.

Now each family $\cB_i$ consists of balls whose radii differ by at most a factor of 2.  In particular, as \textbf{($\star$)} implies that the distance between the center of two balls in $\cB$ is at least the minimum of the radii of the two balls, within each family $\cB_i$, we know that the distance between the centers of two balls is at least half the maximum of the radii of the two balls.  In particular, if we define families $\tilde \cB_i$ by rescaling the balls in each family $\cB_i$ by a factor of $\frac 1 2$,  then each family $\tilde \cB_i$ is a family of disjoint balls.  As such, we have from the condition \eqref{e.ahlfors} that
\begin{equation}\label{e.ballbound}
|\cB_k|\leq C 2^{kd},
\end{equation}
for a fixed constant $C$ depending only on the metric space $\cM$.

In particular, we can bound the total length $L$ of the nearest neighbor tour by the radii $r(B)$ of the balls $B\in \cB$ as follows:
\begin{equation}
L\leq \sum_{B\in \cB} r(B)=\sum_{k\geq 1}\sum_{B\in \cB_k}r(B)\leq C_0\sum_{k=1}^{k_0} 2^{k(d-1)}\leq C_0\frac{2^{k_0d(1-1/d)}}{2^{d-1}-1},
\end{equation}
where $k_0$ is smallest integer for which the bound $C2^{k_0 d}$ on $|\cB_{k_0}|$ from \eqref{e.ballbound} exceeds $n$.  We have thus that for any $d>1$ and a constant $C_1$ depending on the metric space $\cM$ but not the point set $X$, that
\[
L\leq C_1 n^{1-\frac 1 d},
\]
proving the theorem.
\end{proof}

\end{document}